\documentclass[11pt]{article}

\usepackage[paper=a4paper,top=3.0cm,bottom=3.5cm,left=3.cm,right=3.0cm]{geometry}

\usepackage{amsmath,amssymb,amsthm}
\usepackage{xspace}
\usepackage{mathrsfs}
\usepackage{graphicx}
\usepackage{color}
\usepackage[T1]{fontenc}
\usepackage{lmodern}
\usepackage{enumitem}
\usepackage{tabularx}

\usepackage{todonotes}

\usepackage[pdftex,linktocpage,plainpages=false,breaklinks,hypertexnames=false]{hyperref}

\newtheorem{theorem}{Theorem}
\newtheorem{definition}[theorem]{Definition}
\newtheorem{lemma}[theorem]{Lemma}

\newtheorem{corollary}[theorem]{Corollary}

\newcommand{\RR}{\mathbb{R}}

\renewcommand{\Pr}[2][]{\mbox{\rm\bf Pr}_{#1}\hspace{-0.03cm}\left[#2\right]}
\newcommand{\eps}{\varepsilon}

\newcommand{\diam}{\text{diam}}
\newcommand{\cen}{\text{cen}}
\newcommand{\cost}{\text{cost}}
\newcommand{\opt}{\text{opt}}

\newcommand{\kk}{\texorpdfstring{$k$}{k}}

\usepackage{thm-restate}

\title{Fully dynamic hierarchical diameter $k$-clustering and \texorpdfstring{$k$}{k}-center}
\author{Melanie Schmidt\\
  University of Bonn\\
  Germany
  \and
  Christian Sohler \\
  Google Research \\
  Switzerland
}

\begin{document}

\maketitle

\begin{abstract}
  We develop dynamic data structures for maintaining a hierarchical $k$-center clustering when the points come from a discrete space $\{1,\dots,\Delta\}^d$.
  Our first data structure is for the low dimensional setting, i.e., $d$ is a constant, and processes insertions, deletions and cluster representative queries in
  $\log^{O(1)} (\Delta n)$ time, where $n$ is the current size of the point set.
  For the high dimensional case and an integer parameter $\ell \ge 1$, we provide a randomized data structure that maintains an $O(d\ell)$-approximation. The amortized expected insertion time
  is $O(d^2 \ell \log n \log \Delta)$. The amortized expected deletion time is $O(d^2 n^{1/\ell} \log^2 n \log \Delta)$. At any point of time, with probability at least $1-1/n$ the data
  structure can correctly answer all queries for cluster representatives in $O(d\ell \log n \log \Delta)$ time per query.
\end{abstract}

\section{Introduction}

Clustering is a well-studied area in the intersection between combinatorial optimization, algorithms and theory, and machine learning. It seeks to find structure in data by grouping data points into clusters and trying to optimize
a given objective function that defines the clustering problem at hand. One of the best studied clustering problems is the $k$-center problem: Given a set of points $P$ in a metric space, choose $k$ centers such that the maximum
distance of any point to its closest center is minimized. The approximability of $k$-center is well-studied. It is NP-hard to find a $(2-\eps)$-approximation for any constant $\epsilon >0$~\cite{HN79}, and two elegant $2$-approximations are
known~\cite{G85,HS86}.

Gonzalez~\cite{G85} picks the first center arbitrarily and then always chooses a point at maximum distance to the previously chosen centers. This is called \emph{farthest first traversal} and a short proof shows that it gives a $2$-approximation.
Hochbaum and Shmoys~\cite{HS86} give a slightly more complicated, but also more versatile algorithm. At first they observe that one can guess the optimum value since it has to be one of the $O(n^2)$ pairwise distances between the points in $P$. Then they give an algorithm for one guess $\tau$. It computes a maximal independent set $I$ in the graph where two points are connected iff their distance is at most $2\tau$. If $|I|\le k$, then $I$ defines a $k$-center solution of cost $2\tau$. And if $\tau$ was guessed correctly, then $|I|$ cannot be larger than $k$: In a $k$-center solution with radius $\tau$, two points have distance $> 2\tau$ only if they are in different clusters, so there cannot be more than $k$ points with pairwise distance $> 2 \tau$.

A close relative of the $k$-center problem is the diameter $k$-clustering problem that is defined without the need for centers: Given a set of points $P$ in a metric space, find a partitioning of $P$ into $k$ clusters such that the maximum diameter is minimized. 
Both algorithms above can be adapted to give a $2$-approximation for this problem as well. 

The complexity of approximating these problems has been determined in the eighties, yet a lot of work was built on top of those basic results.
Indeed, the $k$-center problem seems to function as a first step in attacking new challenges for clustering problems in general.
Recent work includes solving capacitated clustering where clusters or centers have an upper bound on the number of points that can be assigned (first results gave a $6$-approximation for uniform capacities~\cite{BKP93,KS00},
best known so far is a $9$-approximation for arbitrary capacities~\cite{ABCGMS15}), a $2$-approximation has been obtained when there are lower bounds on the number of points per cluster \cite{CGK16}, and when allowing outliers \cite{CGK16}. 
Other constraints under which the $k$-center problem has been studied include fault tolerance~\cite{KPS00}, matroid or knapsack constraints~\cite{CLLW16}, diversity~\cite{LYZ10} and fairness~\cite{CKLV17}.

Here, we are most interested in a different type of challenge: \emph{Hierarchical} clustering. A hierarchical clustering consists of $n$ partitionings of $P$ that are \emph{nested}, i.e., for any $1 < i \le n$, the $(i-1)$-partitioning results from the
$i$-partitioning by merging two of the clusters. Hierarchical clustering is a popular data analysis method and there is a simple greedy algorithm called agglomerative clustering for it: Starting with $n$ singleton clusters,
successively merge the two clusters that result in a cluster of minimum radius. This is a greedy algorithm since it minimizes the radius in the next step. 

It is not even completely clear how we should evaluate the quality of a hierarchical clustering, yet one possibility is to compute its \emph{pointwise} approximation ratio by comparing the $i$-clustering with an optimal $i$-clustering for any $i\in [n]$ and taking the maximum. 
Dasgupta and Long~\cite{DL05} were the first to study hierarchical $k$-center in this model. They show that agglomerative clustering is at best $\Omega(\log k)$-pointwise approximate. Furthermore, they observed that the algorithm by Gonzalez computes an incremental clustering (centers are only added when aiming for higher $k$), but not a hierarchical clustering and develop an extension of the algorithm that computes an $8$-pointwise approximate hierarchical $k$-center clustering. This is still the best result known; a newer different line of work for $k$-median~\cite{LNRW10} can be applied to $k$-center, too, yet also yields an $8$-pointwise approximation.

At first glance, a hierarchical clustering is a cumbersome object: It consists of $n$ different clusterings, and thinking of the algorithm by Gonzalez, inserting a point or deleting a point could change the incremental clustering completely (and thus, also the hierarchical clustering computed by Dasgupta and Long). In this work, we make a surprising observation: It is possible to design a \emph{fully dynamic} data structure for the $k$-center problem, if the underlying metric space is the Euclidean space $\mathbb{R}^d$.
Our data structure supports insertions, deletions and queries for cluster membership that return a cluster representative. 
To achieve this, we build on the $k$-center algorithm in the streaming setting by McCutchen and Khuller~\cite{MK08} (which in turn implicitly builds upon the algorithm by Hochbaum and Shmoys \cite{HS86}). This streaming algorithm does not explicitly compute a hierarchical clustering, yet we observe in Section~\ref{sec:truncateddendrogram} how to use ideas from ~\cite{MK08} to get a hierarchical clustering.
For this we define a concept of $\alpha$-good set families which basically define a truncated
hierarchical clustering which is pointwise approximate in Section~\ref{sec:truncateddendrogram}. Then the main contribution of this paper is that we maintain such good set families in dynamic settings. First, we consider the $k$-center problem and the diameter $k$-clustering problem in Euclidean space of low dimension, i.e., for constant $d$, and achieve the following result in Section~\ref{sec:lowdim}, where we assume that the points come from the discrete space $\{1,\dots,\Delta\}^d$. The latter assumption can be seen as a way to phrase the quality
in terms of the \emph{spread} of the point set, i.e., the quotient of the largest and smallest pairwise distance.

\begin{restatable}{theorem}{thmLowDim}
  Let $d$ be a constant. There is a data structure for the hierarchical $k$-center or diameter $k$-clustering problem when points come from the discrete space $\{1,\dots,\Delta\}^d$
  that maintains a $16$-approximation with expected amortized insertion time $\mathcal{O}_d(\log \Delta \log n)$ and expected amortized  deletion time $\mathcal{O}_d(\log^2\Delta\log n)$.
  At any point of time, we can query the cluster representative of any point in the data structure in time $\mathcal{O}_d(\log \Delta+\log n)$ time.
\end{restatable}

As a comparison, Chan et al.~\cite{CGS18} design a fully dynamic algorithm for $k$-center, also by extending the algorithm by Hochbaum and Shmoys. The algorithm  works for a pre-specified $k$ given in advance and computes a $(2+\epsilon)$-approximation for the $k$-center problem. With constant probability, it has an amortized update time of $\mathcal{O}(k^2 \epsilon^{-1} \log \Delta )$. 
In constant dimension, our algorithm achieves competitive (slightly worse) update times, yet maintains a complete clustering hierarchy in this time instead of only one clustering.
We then proceed to the high dimensional case, showing the following result in Section~\ref{sec:highdim}.

\begin{restatable}{theorem}{thmHighDim}
  Let $\ell \ge 1$ be an integer. 
  There is a data structure for the hierarchical diameter clustering that maintains an $O(d\ell)$-approximation under insertions and deletions of points from $\{1,\dots,\Delta\}^d$. The amortized expected insertion time is
  $O(d^2 \ell \log n \log \Delta)$. The amortized expected deletion time is $O(d^2 n^{1/\ell} \log^2 n \log \Delta)$. At any point of time, with probability at least $1-1/n$ the data structure can correctly answer all queries 
  for cluster representatives in $O(d\ell \log n \log \Delta)$ time per query.
\end{restatable}

In this case, the deletion time (but not the insertion time) is no longer polylogarithmic in $\Delta$. Yet we can decrease the dependence on $n$ to $n^\epsilon$ for any constant $\epsilon$, still maintaining a constant-factor approximation for every level in the hierarchy.
Also, our result implies that a hierarchical clustering can be computed in time $O(n d \log n \log \Delta)$ (using only insertion operations), which is much faster than for example the algorithm by Dasgupta and Long (even computing the incremental clustering in the
beginning takes time $O(n^2 d)$).

\paragraph{Additional related work.} In this paper, we consider \emph{Euclidean} $k$-center where the underlying metric is the Euclidean space; in this case, the lower bounds on the approximation ratios for $k$-center and diameter-$k$-clustering are not completely tight,
but it is known that finding a $1.82$-approximation for the $k$-center problem and a $1.96$-approximation for the diameter $k$-clustering problem is NP-hard~\cite{FG88}. 

We will discuss the streaming algorithm by McCutchen and Kuhller in more detail in Section~\ref{sec:truncateddendrogram}. It provides an $8$-approximation for the $k$-center problem for on fixed $k$ in the streaming setting, while only storing at most $k+1$ points at any point in time.

Cohen-Addad et al.~\cite{CASS16} study the $k$-center problem in a different streaming setting, the sliding window model. In this model,
the goal is to maintain a solution that is always an approximation for the $N$ most recent points in the stream. 
For the metric diameter problem, they give a $(3+\epsilon)$-approximation that stores $\mathcal{O}(\epsilon^{-1}\log \Delta)$ points and updates in time $\mathcal{O}(\epsilon^{-1} \log \Delta)$, and for the $k$-center problem, they provide a $(6+\epsilon)$-approximation
storing $\mathcal{O}(k \epsilon^{-1} \log \Delta)$ points and updates in time $\mathcal{O}(k^2 \epsilon^{-1} \log \Delta)$. The number of clusters $k$ has to be specified in advance.

There is also a line of work considering the running time and space complexity of agglomerative clustering. 
The starting point of this is that in its standard form, agglomerative clustering requires $O(n^2)$ time and space.
Eppstein~\cite{E00} shows that agglomerative clustering can be performed in space $O(n)$ in time $O(n^2 \log^2 n)$. 
Cochez and Mou~\cite{CM15} and Gilpin et al.~\cite{GQD13} achieve an approximate clustering in the sense that it is an
approximation of agglomerative clustering. They achieve a linear time algorithm which requires linear space.

\subsection{Preliminaries}

We assume that our points are from $\{1,\dots,\Delta\}^d\subset \mathbb{R}^d$. We aim at algorithms whose running time is polylogarithmic in $\Delta$ and $n$ (we allow for multiple insertions of a point). 

We consider two different settings: In a low-dimensional setting, we
will assume $d$ to be a constant. This is relevant when we evaluate the update time of our data structure, as in this case the $O$-notation will swallow any function that only depends on $d$. We will make this
clear by writing $O_d()$ to denote the $O$-notation when $d$ is considered to be a constant. In the high-dimensional case, $d$ is not considered to be constant. We may, however, assume that $d=O(\log n)$ as we
can use the Johnson-Lindenstrauss lemma to embed the points in $O(nd\log n)$ time in $O(\log n)$ dimensions in such a way that distances are preserved by a small constant. 

\subsection{The Hierarchical Diameter \kk-Clustering Problem}

We start by defining the diameter $k$-clustering problem. Let $P \subseteq \{1,\dots,\Delta\}^d$ be a point set. The diameter $\diam(C)$ of a subset $C\subseteq P$ is defined as $\diam(C) = \max_{p,q \in C} \|p-q\|_2$.
The goal of the diameter $k$-clustering problem is to find a partition of the input point set $P$ into $k$ subsets $C_1,\dots, C_k$ such that the maximum diameter of the $k$ subsets is minimized.
We will write a partition of $P$ as a set $\mathcal C$ that contains $k$ disjoint sets $C_1,\dots,C_k$ whose union is $P$. With this definition our objective is to minimize
$$
\cost(P,\mathcal C)= \cost_{\diam}(P,\mathcal C) = \max_i \diam(C_i).
$$
A partition $\mathcal C =\{C_1,\dots,C_k\}$ is a refinement of a partition $\mathcal D=\{D_1,\dots,D_\ell\}$, if for every $C_i$ there is one $D_j$ such that $C_i \subseteq D_j$. If $\ell< k$ then we say that $\mathcal C$ is
a proper refinement of $\mathcal D$. A hierarchical clustering of a points set $P =\{p_1,\dots,p_n\}$ is a sequence of partitions $\mathcal C_1,\dots, \mathcal C_n$ such that $\mathcal C_1 =\{P\}$, $\mathcal C_n = \big\{\{p_1\},\dots,\{p_n\}\big\}$
and such that every $\mathcal C_i$ is a proper refinement of $\mathcal C_{i-1}$ that is obtained by splitting one cluster of $\mathcal C_{i-1}$ into two clusters and keeping the other clusters unchanged.
We can describe such a hierarchical clustering by a binary tree: The leaves of the tree are the points of the input point set and the root corresponds to the whole point set $P$. Every inner node corresponds to a cluster
that contains all points located in its subtree. An inner node $v$ of the tree is labeled by the first index $i$ such that the partition $\mathcal C_i$ does not contain the nodes in the subtree of $v$ in a single cluster.
Such a tree is also called a \emph{dendrogram}.

We denote the diameter of a optimal partition into $k$ clusters by $\opt_{\diam}^k$. 
A partition $\mathcal C_A$ is an $\alpha$-approximation to the diameter $k$-clustering problem if $\cost(P,\mathcal C_A) \le \alpha \cdot \opt_{\diam}^k$.
A hierarchical clustering $\mathcal C_1,\dots,\mathcal C_n$ is a \emph{pointwise} $\alpha$-approximation to the hierarchical $k$-clustering problem if for \emph{every} $k, 1\le k \le n,$ we have that
$\mathcal C_k=\{C_1,\dots,C_k\}$ is an $\alpha$-approximation to the diameter $k$-clustering problem, i.e., we have a hierarchical clustering that in each step is approximately as good as the best non-hierarchical clustering. 

\paragraph{Relations to the \kk-center problem}
The $k$-center problem is to find a set $\mathcal{C}\subseteq P$ of $k$-centers that minimize the
maximum distance to the nearest center, i. e., that minimizes 
\[
\cost_{\cen}(P,\mathcal{C}) = \max_{p\in P} \min_{c\in C} \|p-c\|_2.
\] 
By definition, we know that $\cost_{\cen}(P,\mathcal{C}) \le \cost_{\diam}(P,\mathcal{C}) \le 2 \cdot \cost_{\cen}(P,\mathcal{C})$.
The same is also true when the centers are chosen from $\mathbb{R}^d$ instead of $P$.
We denote the cost of an optimal $k$-center clustering with $k$ centers by $\opt_{\cen}^k$.

For the $k$-center problem, a solution consists of centers, and the clusters are induced by assigning points to their closest center. For the diameter $k$-clustering problem, a solution is a partition, and no centers are associated. However, our data structure maintains a special point for each clusters in both cases. We adopt the view of diameter $k$-clustering in the following, and therefore speak of \emph{representatives} of clusters. When viewed as a $k$-center clustering, the representatives serve as centers. In general, our exposition is focused on the $k$-diameter case, while only making slight distinctions where necessary for $k$-center.

\section{Pointwise approximate hierarchical clusterings}\label{sec:truncateddendrogram}

First we define the hierarchical clustering that we want to maintain.
%
%
Let us recall the $2$-approximation algorithm for $k$-center due to Hochbaum and Shmoys~\cite{HS86}. Given a point set $P$ and a guess $\tau$ for the optimum radius, this algorithm computes a maximal independent set $I$ in the graph where two points are connected iff their distance is at most $2\tau$. If $|I|\le k$, then it defines a $k$-center solution of cost $2\tau$. And if $\tau$ was guessed correctly, then $|I|$ cannot be larger than $k$: In a $k$-center solution with radius $\tau$, two points have distance $> 2\tau$ only if they are in different clusters, so there cannot be more than $k$ points with pairwise distance $> 2 \tau$.

We adopt an idea from McCutchen and Khuller~\cite{MK08}. The paper~\cite{MK08} obtains a streaming $8$-approximation for $k$-center for one predetermined value of $k$. The key idea to make this work is to maintain a lower bound $\ell$ on the optimum cost and relate this to the cost of the solution with $\le k$ centers that they keep in memory. 
Whenever $k+1$ points with pairwise distance at least $\ell$ have been found, $\ell$ is doubled and a maximum independent set is computed in a very similar way to the algorithm of Hochbaum and Shmoys. 
We observe that performing a similar strategy (in a non-streaming setting) can be used to obtain a hierarchical clustering which is a pointwise $8$-approximation.

More precisely, we show that we can compute a nested family of subsets of $P$ that satisfies the following conditions, and that this family induces a pointwise approximate hierarchical clustering.

\begin{definition}\label{def:truncated}
For a point set $P \in \{1,\ldots,\Delta\}^d$, we call a nested family $(P_i)_{i=0}^M$, $M:= \log (\sqrt{d}\cdot(\Delta-1))$, \emph{$\alpha$-good} if the following conditions hold:
\begin{itemize}
\item[(a)]
  $P_0 = P$, $P_i \subseteq P_{i-1}$, and  $P_M = \{p\}$ for some $p \in P$,
\item[(b)]
  Every pair of distinct points $p,q \in P_i$ has distance more than $2^i$,
\item[(c)]
  For $1\le i \le M$, every point in $P_{i-1} \setminus P_i$ has distance at most $\alpha\cdot 2^{i}$ to some point in $P_i$.
\end{itemize}
\end{definition}

In a static setting, we can compute a $1$-good family $(P_1)_{i=0}^M$ in the following way. 
We start with $P_0 = \{\{x\}\mid x \in P\}$. This satisfies condition (a) by definition, it satisfies condition (b) for $P_0$ since all points have a pairwise distance of at least $1=2^0$, and condition (c) is not applicable to $P_0$.

Now for any $i\in \{1,\ldots,M\}$, we construct $P_{i}$ from $P_{i-1}$: We greedily compute a maximum independent set in a graph $G_{i-1}$ with vertex set $P_{i-1}$ where we connect two points $x,y$ iff $||x-y||_2 \le 2^{i}$, and set $P_{i}$ be this independent set. By definition, $P_{i} \subseteq P_{i-1}$, i.e., condition (a) holds. Furthermore, any two points in $P_{i}$ are independent in $G_i$, which means that their distance is more than $2^{i}$. Thus, condition (b) is true for all points in $P_{i}$. Finally, every point in $P_{i-1}$ that was not promoted to $P_{i}$ has a neighbor in $P_{i}$, i.e., a point at distance $2^{i}$. Thus, condition (c) is true for all points in $P_{i-1} \setminus P_{i}$.

This process ends at $P_M$: The largest possible distance between two points is $\sqrt{d}\cdot(\Delta-1)$. Thus, if we pick an arbitrary point, then any other point is at distance $\le \sqrt{d}\cdot(\Delta-1)=2^M$. Thus, by condition (b), $P_M$ can not contain more than one point, and this is why we stop at this level (and why it is well-defined to demand that $P_M$ contains only one point).

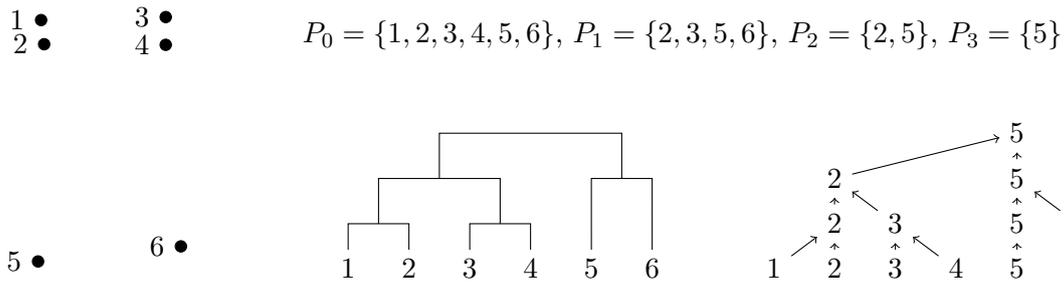
\begin{figure}
\centering
\begin{tikzpicture}[scale=0.8]
\node [circle, draw, inner sep=0cm, minimum size=0.15cm, fill,label=left:{$1$}] at (-0.05,0) {};
\node [circle, draw, inner sep=0cm, minimum size=0.15cm, fill,label=left:{$2$}] at (0,-0.4) {};
\node [circle, draw, inner sep=0cm, minimum size=0.15cm, fill,label=left:{$3$}] at (2,0.05) {};
\node [circle, draw, inner sep=0cm, minimum size=0.15cm, fill,label=left:{$4$}] at (2,-0.41) {};
\node [circle, draw, inner sep=0cm, minimum size=0.15cm, fill,label=left:{$5$}] at (-0.1,-4) {};
\node [circle, draw, inner sep=0cm, minimum size=0.15cm, fill,label=left:{$6$}] at (2.25,-3.75) {};


\node at (10.5,-0.25) {$P_0=\{1,2,3,4,5,6\}$, $P_1=\{2,3,5,6\}$, $P_2=\{2,5\}$, $P_3=\{5\}$};

\begin{scope}[yscale=1.5,yshift=1.25cm,xshift=12cm]
\node (b11) at (0,-4) {$1$};
\node (b21) at (1,-4) {$2$};
\node (b31) at (2,-4) {$3$};
\node (b41) at (3,-4) {$4$};
\node (b51) at (4,-4) {$5$};
\node (b61) at (5,-4) {$6$};

\begin{scope}[yshift=0.5cm]
\node (b22) at (1,-4) {$2$};
\node (b32) at (2,-4) {$3$};
\node (b52) at (4,-4) {$5$};
\node (b62) at (5,-4) {$6$};
\end{scope}

\begin{scope}[yshift=1cm]
\node (b23) at (1,-4) {$2$};
\node (b53) at (4,-4) {$5$};
\end{scope}

\begin{scope}[yshift=1.5cm]
\node (b54) at (4,-4) {$5$};
\end{scope}

\draw [->] (b11) -- (b22);
\draw [->] (b21) -- (b22);
\draw [->] (b31) -- (b32);
\draw [->] (b41) -- (b32);
\draw [->] (b51) -- (b52);
\draw [->] (b61) -- (b62);

\draw [->] (b22) -- (b23);
\draw [->] (b32) -- (b23);
\draw [->] (b52) -- (b53);
\draw [->] (b62) -- (b53);

\draw [->] (b23) -- (b54);
\draw [->] (b53) -- (b54);

\end{scope}
\begin{scope}[yscale=1.5,yshift=1.25cm,xshift=5cm]
\node (b1) at (0,-4) {$1$};
\node (b2) at (1,-4) {$2$};
\node (b3) at (2,-4) {$3$};
\node (b4) at (3,-4) {$4$};
\node (b5) at (4,-4) {$5$};
\node (b6) at (5,-4) {$6$};
\draw (b1) -- (0,-3.5) -- (1,-3.5) -- (b2);
\draw (b3) -- (2,-3.5) -- (3,-3.5) -- (b4);
\draw (0.5,-3.5) -- (0.5,-3) -- (2.5,-3) -- (2.5,-3.5);
\draw (b5) -- (4,-3) -- (5,-3) -- (b6);
\draw (1.5,-3) -- (1.5,-2.5) -- (4.5,-2.5) -- (4.5,-3);
\end{scope}
\end{tikzpicture}
\caption{A point set $P$, an $\alpha$-good set family, and two ways to represent this family as a tree. The middle one shows a compacted dendrogram, where multiple merges can happen at the same time. The right one shows a representation that is closer to the set family.\label{fig:truncateddendrogram}}
\end{figure}

An $\alpha$-good set family has two beneficial properties. 
\paragraph{Number of levels.} 
Firstly, it has exactly $M+1 \in \mathcal{O}(\log \Delta)$ levels\footnote{Notice that the $P_i$ in our definition are not necessarily different. It may well be that there is a $P_i$ where all points already have pairwise distance $2^{i+1}$ (or even higher), and that we then have $P_{i+1}=P_i$ (or even multiple identical levels). This could be removed by further condensing, yet we prefer this version since it simplifies the exposition. Since we do not remove identical levels, the number of levels of the tree is always exactly $M+1$.}. 
This is beneficial for maintaining the family in short update time. Figure~\ref{fig:truncateddendrogram} shows an example where we have a point set, an $\alpha$-good set family, and two tree structures based on it. The middle representation shows a variant of a dendrogram where multiple merges can happen at the same \lq height\rq, when multiple points are left out from one $P_{i-1}$ to $P_i$. This visualizes how the dendrogram gets compressed such that $M$ levels are sufficient.

\paragraph{Approximate Clusterings.}
Secondly, an $\alpha$-good family implies a pointwise approximate hierarchical clustering. We prove this statement in the following. Assume that we have access to a function $p(i,x)$ which for any $i \in \{1,\ldots,M\}$ and any $x \in P_{i-1}$ gives us a point in $P_i$ at distance at most $\alpha 2^{i}$. We call $p(i,x)$ the \emph{parent} of $x$ on level $i$. This corresponds to the tree structure on the right in Figure~\ref{fig:truncateddendrogram}.
Next we define a function $p^i(x)$ which gives the representative of the cluster that $x$ belongs to in the clustering represented by $P_i$. We get this function by setting $p^0(x)=x$, $p^1(x) = p(1,x)$ and then recursively setting $p^i(x) = p(i,p^{i-1}(x))$ for $i=\{2,\ldots,M\}$. Notice that $p^i(x) \in P_i$. 
%
The appealing property of compacted dendrograms is that $p^i(x)$ is always close to $x$ compared to the distance lower bound associated to level $i$. More precisely, the following is true:

\begin{lemma}\label{lem:distToParent}
  For every point $x \in P$, $\|x-p^{i}(x)\|_2 \le \alpha \cdot 2^{i+1}$.
\end{lemma}
\begin{proof}
  By triangle inequality and the fact that by definition of $p(z,j)$, $\|z-p(z,j)\|_2 \le \alpha \cdot 2^{j}$ for all $j \in \{1,\ldots,M\}$, $z\in P_{j-1}$, we have that
	\begin{align*}
	\|x-p^i(x)\|_2 
	\le \|x-p^1(x)\|_2 + \|x-p^2(x)\|_2+ \ldots \|x-p^i(x)\|_2
	\le& 2^1 + 2^2 + \ldots +2^{i}\\
	=& \sum_{j=0}^i \alpha \cdot 2^{j} 
	\le \alpha \cdot 2^{i+1}.\qedhere
	\end{align*}
\end{proof}

Now we set $C_i(z) := \{x \in P \mid p^i(x) = z\}$ for all $i \in \{0,\ldots,M\}$ and all $z \in P_i$. This is the cluster represented by the copy of $z$ on level $i$: It contains all unique points in the subtree rooted at $z$ on level $i$. This can change with the level: For example,  in Figure~\ref{fig:truncateddendrogram}, $C_0(2)=\{2\}$, $C_1(2) = \{1,2\}$ and $C_2(2) = \{1,2,3,4\}$ (and $C_3(2)$ is not defined). 
For every $i$, we get a clustering with $|P_i|$ clusters $\mathcal{C}_i = \{ C_i(x) \mid x \in P_i\}$. For $k$-center we choose $P_i$ as the center set for this clustering.

We prove that for every $i$, $\mathcal{C}_i$ is an $8\alpha$-approximation for the best diameter $k$-clustering solution with $k=|P_i|$ clusters. Even more, we show that $\mathcal{C}_i$ is also a good clustering for all $k$ that are larger than $|P_i|$, but strictly smaller than $|P_{i-1}|$. This is because $|P_{i-1}|$ already constitutes the lower bound which makes $\mathcal{C}_i$ a good clustering in comparison.

\begin{lemma}\label{lem:approxguarantee}
Let $i \in \{1,\ldots,M\}$ with $|P_i| < |P|$, and let $j \in \{0,\ldots,i-1\}$ be the largest index for which $|P_j| > |P_i|$. 
Then $\cost(P,\mathcal{C}_i) \le 8 \alpha \cdot \opt^{k}_{\diam}(P)$ and $\cost_{\cen}(P,P_i) \le 8 \alpha \cdot \opt_{\cen}^{k}(P)$ for all $k$ with $|P_i| \le k < |P_j|$.
\end{lemma}
\begin{proof}
Fix an arbitrary $k$ with $|P_i| \le k < |P_j|$. Notice that $j < i$ by definition, and that $j$ is well-defined since $|P_i| < |P|$ and $|P_0|=|P|$.
Since $|P_j| > k$, we know by property $(b)$ of Definition \ref{def:truncated} that there are at least $k+1$ points of pairwise distance at least $2^j$ in $P$.
In every $k$-clustering $\mathcal C^*$, there must be a cluster that contains two of these points. The diameter of this cluster is at least $2^j$, so $\opt^k_{\diam}(P) \ge 2^j$,
and the radius is at least $2^{j-1}$, so $\opt^k_{\cen}(P) \ge 2^{j-1}$.

We know by property (a) that $P_j \subseteq P_{j+1} \subseteq \ldots \subseteq P_i$, and since $j$ is the largest index with $|P_j| > |P_i|$, we know in particular that $P_{j+1} = P_i$.

Furthermore, by Lemma~\ref{lem:distToParent}, we know that for every $x \in P$ it holds that $\|x-p^{j+1}(x)\|_2 \le \alpha \cdot 2^{j+1+1}$. Since $p^{j+1}(x) \in P_{j+1}$, this means
that every point in $P$ has a point at distance at most $\alpha \cdot 2^{j+2}$ in $P_{j+1}=P_i$. Thus, using $P_i$ as a center set for $k$-center yields a solution of radius
$\alpha \cdot2^{j+2} \le 8 \alpha \opt_{\cen}^k(P)$. 
Also, the distance between any pair of points having the same parent in $P_i$ is at most $\alpha \cdot2^{j+3}$, so if we define that a cluster consists of all points having the same parent,
we get $|P_i|$ clusters with a maximum diameter of $\alpha \cdot2^{j+3}\le 8\alpha \cdot \opt^k_{\diam}(P)$.
\end{proof}

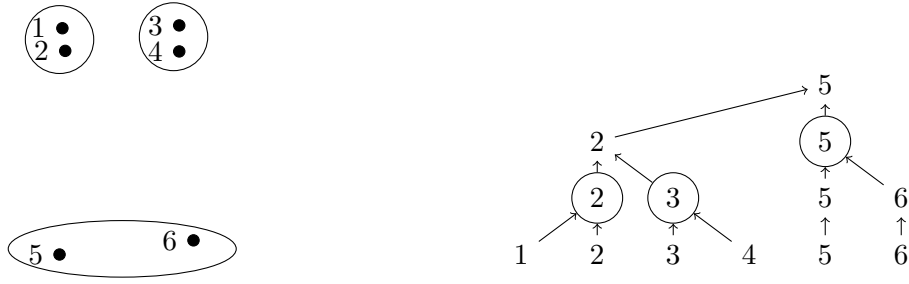
\begin{figure}
\centering
\begin{tikzpicture}
\begin{scope}[scale=0.75]
\node [circle, draw, inner sep=0cm, minimum size=0.15cm, fill,label=left:{$1$}] at (-0.05,0) {};
\node [circle, draw, inner sep=0cm, minimum size=0.15cm, fill,label=left:{$2$}] at (0,-0.4) {};
\node [circle, draw, inner sep=0cm, minimum size=0.15cm, fill,label=left:{$3$}] at (2,0.05) {};
\node [circle, draw, inner sep=0cm, minimum size=0.15cm, fill,label=left:{$4$}] at (2,-0.41) {};
\node [circle, draw, inner sep=0cm, minimum size=0.15cm, fill,label=left:{$5$}] at (-0.1,-4) {};
\node [circle, draw, inner sep=0cm, minimum size=0.15cm, fill,label=left:{$6$}] at (2.25,-3.75) {};

\draw (-0.1,-0.2) circle (0.6cm);
\draw (1.9,-0.15) circle (0.6cm);
\draw (1,-3.9) ellipse (2cm and 0.5cm);
\end{scope}

\begin{scope}[yscale=1.5,yshift=2cm,xshift=6cm]
\node (b11) at (0,-4) {$1$};
\node (b21) at (1,-4) {$2$};
\node (b31) at (2,-4) {$3$};
\node (b41) at (3,-4) {$4$};
\node (b51) at (4,-4) {$5$};
\node (b61) at (5,-4) {$6$};

\begin{scope}[yshift=0.5cm]
\node [draw, circle]  (b22) at (1,-4) {$2$};
\node [draw, circle]  (b32) at (2,-4) {$3$};
\node (b52) at (4,-4) {$5$};
\node (b62) at (5,-4) {$6$};
\end{scope}

\begin{scope}[yshift=1cm]
\node (b23) at (1,-4) {$2$};
\node [draw, circle] (b53) at (4,-4) {$5$};
\end{scope}

\begin{scope}[yshift=1.5cm]
\node (b54) at (4,-4) {$5$};
\end{scope}

\draw [->] (b11) -- (b22);
\draw [->] (b21) -- (b22);
\draw [->] (b31) -- (b32);
\draw [->] (b41) -- (b32);
\draw [->] (b51) -- (b52);
\draw [->] (b61) -- (b62);

\draw [->] (b22) -- (b23);
\draw [->] (b32) -- (b23);
\draw [->] (b52) -- (b53);
\draw [->] (b62) -- (b53);

\draw [->] (b23) -- (b54);
\draw [->] (b53) -- (b54);

\end{scope}

\end{tikzpicture}
\caption{The example from Figure~\ref{fig:truncateddendrogram}, and a clustering for $k=3$ which results from assuming that $3$ is the first in the ordering of $P_1 \backslash P_2$. The resulting partition is $\{1,2\},\{3,4\},\{5,6\}$. If $6$ was first in the ordering, we would get $\{1,2,3,4\}, \{5\}, \{6\}$.\label{fig:howtheclusteringscomeabout}
}
\end{figure}
Finally, let us discuss how to obtain clusterings with exactly $k$ clusters for any $k$.
Let $k$ be a fixed number that is strictly between $|P_{i-1}|$ and $|P_i|$ for some $i$.
We know by Lemma~\ref{lem:approxguarantee} that $P_i$ induces a good clustering for this $k$. So one way of promoting this to a clustering with $k$ clusters would be to simply add $k-|P_i|$ points from $P\backslash P_i$ as singleton clusters which can only decrease the cost.
However, this would not lead to a hierarchical clustering anymore. 
We therefore opt for another way: We choose $k-|P_i|$ points $T$ from $P_{i-1} \backslash P_i$, based on a fixed ordering of the points. Then we add $C_{i-1}(y)$ to the clustering for all $y\in T$ while removing the points in $\cup_{y \in T} C_{i-1}(y)$ from the other clusters. Let $\mathcal{C}_i^T$ be the resulting clustering, and set $P_i^T = P_i \cup T$. An example for this is depicted in Figure~\ref{fig:howtheclusteringscomeabout}.

\begin{corollary}
Let $i \in \{1,\ldots,M\}$ with $|P_i| < |P|$, and let $j \in \{0,\ldots,i-1\}$ be the largest index for which $|P_j| > |P_i|$. Let $k$ be such that $|P_i| \le k < |P_j|$ and let $T \subseteq P_j \backslash P_i$ be a set with $|T|=k-|P_i|$.
Then $\cost(P,\mathcal{C}_i^T) \le 8 \alpha \cdot \opt^{k}_{\diam}(P)$ and $\cost_{\cen}(P,P_i^T) \le 8 \alpha \cdot \opt_{\cen}^{k}(P)$.
\end{corollary}
\begin{proof}
  This follows from the proof of Lemma~\ref{lem:approxguarantee}: The lower bound on $\opt^k_{\cen}(P)$ and $\opt^k_{\diam}(P)$ is unchanged, radius / diameter of the clusters in $\mathcal{C}_i$ can only decrease by removing points, and the radius / diameter
  of the new clusters is even smaller, since Lemma~\ref{lem:distToParent} gives a better upper bound for representatives stemming from $P_j$.
\end{proof}

\section{A Data Structure for Points in Low-Dimensional Space}\label{sec:lowdim}

The goal of this section is to present a dynamic data structure which supports the following operations.

\begin{itemize}
\item
  {\sc Insert}$(p)$: A new point $p\in \{1,\dots,\Delta\}^d$ is inserted into $P$. We allow multiple copies of $p$.
\item
  {\sc Delete}$(p)$: A point $p \in \{1,\dots,\Delta\}^d$ is deleted from $P$. If $p$ is not present in the current set $P$, this will be reported.
\item
  {\sc Cluster}$(p,k)$: A representative of the cluster $C_i$ that contains $p$ in the $k$-clustering $\mathcal C_k$ is returned.
\end{itemize}

\paragraph{General structure.}
We realize this data structure by maintaining an $\alpha$-good family of sets $(P_i)_{i=0}^M$ that satisfies the conditions in Definition~\ref{def:truncated} for $\alpha = 2$. We organize these sets in a tree as
shown on the right of Figure~\ref{fig:truncateddendrogram} and \ref{fig:howtheclusteringscomeabout}. 
For this, for every $x\in P_{i-1}\setminus P_i$ we maintain a pointer to a $y=p^i(x)$, more precisely to the copy of $y$ in $P_i$. Notice that $y=x$ is possible, then the pointer just points to the \lq next\rq\ copy of $x$.
Additionally, we also maintain backward pointers, i.e., every $x \in P_i$ has a list of pointers to its children.

Finally, for every $i \in \{1,\ldots,M\}$, we give each point its insertion time as a fixed identifier and use this for the fixed ordering of the points in $P_{i-1} \setminus P_{i}$. We will maintain the point identifiers for the points in $P_{i-1} \setminus P_{i}$
in a binary search tree so we can find points based on the ordering.

\paragraph{Computing the Cluster Representative.}

In order to find the cluster representative of a point we follow its pointers in the above described tree structure until we encounter a vertex that is a cluster representative.  

A point is a cluster representative for a given $k$, if one of the following two things happens: The point is in the set $P_i$ with $|P_i| \le k$ that has the smallest index $i$, or it is one of the $r:=k-|P_i|$
additional points that we choose according to the ordering of the points in $P_{i-1} \setminus P_i$. We can detect the first case when we reach a level where $|P_i|\le k$. For the second case, we need more information.
The binary search tree for every set $P_{i-1} \setminus P_{i}$ allows us to compute the rank of a point in the fixed ordering in $O(\log n)$ time. For every $P_j$ we also maintain its size, so that we can compute $k-|P_i|$ in constant time.
Now we can check whether a point in $P_{i-1} \setminus P_i$ is a cluster representative by comparing the rank to $k-|P_i|$. We choose the first $k-|P_i|$ points to cluster representatives. 

The running time consists of following at most $\mathcal{O}(\log \Delta)$ pointers plus doing potentially one rank computation in $\mathcal{O}(\log n)$,
so we get a running time of $\mathcal{O}_d(\log \Delta + \log n)$. 
\begin{lemma}
The operation {\sc Cluster}$(p,k)$ can be implemented to run in $\mathcal{O}_d(\log \Delta + \log n)$ time. 
\end{lemma}

\paragraph{Insertions.}
In order to describe how insertions are performed, we need to first say how we store the sets $P_i$. For every $0 \le i \le d \log \Delta$ we store $P_i$ in a hash table.\footnote{Using dynamic perfect hashing \cite{DKMMRT94} or cuckoo hashing \cite{PR04}, it is possible to maintain hash maps with constant worst-case search time and constant expected amortized insertion and deletion time.}
The key of a point is the cell of a grid of diameter $2^i/\sqrt{d}$ that contains the point. 
When we insert a point $p$, we start with $P_0$  where the point is always inserted into the hash table. 
Now let $i \ge 1$. 
In order to determine whether $p$ will be inserted in $P_i$ we use our hash table.
We query the hash table for all grid cells that could contain a point in distance $2^i$.
All such cells are covered by a box of width $2\cdot 2^i + 2^i / \sqrt{d}$, so the number of queries to the hash table is bounded by $(2\sqrt{d}+1)^d \in O_d(1)$.
If the hash table is not empty for one of the keys, we take the first point that we find, and add a pointer to it.
Furthermore, we add the point identifier to the binary search tree for $P_{i-1} \backslash P_i$ since the point was not inserted into $P_i$.
If we find no point, then we insert the search point into the hash table for $P_i$ and proceed with the next $i$.

Observe that if we add $p$ to $P_i$, then we did so because we found no point in distance $2^i$, which ensures condition (b) of $\alpha$-good set families. Furthermore, when we add a pointer to a point $q$ in $P_i$
instead of adding $p$, then $q$ is at distance at most $2^i+2/\sqrt{d} = (1+1/(2^{i-1}\sqrt{d}))\cdot 2^i$. So we satisfy condition (c) for $\alpha = (1+1/(2^{i-1}\sqrt{d})) < 2$.

In the worst case, we have to insert the point in $\Theta(\log \Delta)$ levels, and have to insert it into $\Theta(\log \Delta)$ binary search trees, so the worst case insertion time is $O_d(\log \Delta \log n)$.

\begin{lemma}
Insertions can be processed in expected amortized time $O_d(\log \Delta \log n)$.
\end{lemma}

\paragraph{Deletions.}

To delete a point $p$ we first remove it from all $P_i$ and all binary search trees that it is contained in. 
This takes time $\mathcal{O}(\log \Delta + \log \Delta\cdot \log n)$.

Then we need to check whether property (c) is violated. On each level $i$, we iterate through the children of $p$ (this is why we maintain backward edges). For each child point $q$, we continue the insertion process that was interrupted when $p$ was found as a close point on level $i$. This means that we try to insert $q$ into $P_i$, and, if successful, into further levels, just as in the insertion process. The running time for performing this operation is at most the running time of inserting a point normally, i.e., $\mathcal{O}_d(\log \Delta \log n)$.

By the same argumentation as for the insertions, the number of children of a point is $\mathcal{O}_d(1)$. However, since we delete the point from potentially all $\mathcal{O}(\log \Delta)$ levels, we have to process $\mathcal{O}_d(\log \Delta)$ child points in the worst case.

\begin{lemma}
Deletions can be processed in expected amortized time $\mathcal{O}_d(\log^2\Delta \log n)$.
\end{lemma}

Since the family that we maintain is $\alpha$-good for $\alpha < 2$, we get the following result.

\thmLowDim*

\section{In High Dimension}\label{sec:highdim}

We now consider the high dimensional case. We assume that the dimension is $O(\log n)$ since the Johnson-Lindenstrauss lemma allows us to project points to $O(\log n)$ dimensions without
distorting the distances by more than a constant with high probability. In the following, we will assume that the current number of points $n$ is approximately known to the algorithm (say, upto a factor of 2)
and we rebuild the data structure, if this number changes by more than this factor. This does not change the amortized expected cost of the data structure operations. 

We derive a $\alpha$-good family of sets $(P_i)_{i=0}^M$ for $\alpha = 2 d \ell$ according to Definition~\ref{def:truncated}. 
Similar to the low dimensional case, we use pointers for every $p \in P_i$ to a point $q\in P_{i+1}$ that has distance at most $d \ell 2^{i+1}$ from $p$. However, in the high dimensional case, we only maintain the pointers implicitly. They can be obtained from the data
structures storing the $P_i$.

\subsection{Maintaining the \texorpdfstring{$P_i$}{Pi}}

We now describe how to maintain the sets $P_i$ in high dimensions. Let $R=2 \cdot \sqrt{d} \cdot 2^i$. Consider a randomly shifted grid with cell width $R$.
\begin{lemma}
  Let $p = (p_1,\dots,p_d),q=(q_1,\dots,q_d) \in \RR^d$ be two points with distance $\|p-q\|_2 \le 2^i$. 
  Consider a randomly shifted axis-aligned grid with cell width $R = 2 \cdot \sqrt{d}\cdot 2^i$.
  Then 
  $$
  \Pr{\text{$p$ and $q$ are not in the same grid cell}} \le \frac{1}{2}. 
  $$
\end{lemma}
\begin{proof}
  $$
   \Pr{\text{$p$ and $q$ are not in the same grid cell}} \le \sum_{1\le i\le d}\frac{|p_i-q_i|}{R} = \frac{\|p-q\|_1}{R} \le \sqrt{d} \frac{\|p-q\|_2}{R} \le \frac{1}{2}.
  $$
\end{proof}

For each level $i$ we maintain $g=O(\log n)$ shifted grids with side length $R$ as defined in the above lemma. In order to compute $P_i$ from $P_{i-1}$ we first observe
that any pair of points in $P_{i-1}$ with distance at most $2^i$ is with high probability in the same grid cell in at least one of the grids. Therefore, during the computation
of $P_i$ from $P_{i-1}$ we will make sure that there are no two points in $P_i$ that are both in the same grid cell for one of the grids. This will ensure the second
property of our data structure, i.e. every pair of distinct points $p,q\in P_i$ has distance at least $2^i$.

To compute $P_i$ from $P_{i-1}$ we maintain a sequence of subsets $P_{i,0},\dots, P_{i,\ell}$ such that $P_{i,0} = P_{i-1}$ and $P_{i,j}$ is obtained from $P_{i,j-1}$ by sampling every point
from $P_{i,j-1}$ independently and uniformly at random with probability $n^{-\frac{1}{\ell}}$. Furthermore, we define $P_{i,\ell+1} = \emptyset$. We say that a point $p\in P_{i,j}$ is \emph{covered},
if there is a point $q\in P_{i,j+1}$ that is in the same grid cell as $p$ in one of the grids. The remaining points are called \emph{uncovered}. From the uncovered points we select a maximal
subset $I_{i,j}$ such that no two points in $I_{i,j}$ are contained in the same grid cell for some of the grids. Finally, we define $P_i =\bigcup_j I_{i,j}$.

\begin{lemma}
For every $p\in P_{i-1}$ there exists a point $q \in P_i$ such that $\|p-q\|_2 \le \sqrt{d} \cdot R \ell$.
\end{lemma}
\begin{proof}
Let $p \in P_{i-1}$. If $p\in P_i$ we are done. Thus, let us assume that $p\notin P_i$. In this case, either $p$ is covered or uncovered. If $p$ is uncovered, by definition of the $I_{i,j}$
there is a point $q\in P_i$ within distance $\sqrt{d} \cdot R$. Thus, let us assume $p$ is covered and $p\in P_{i,j}$. Then we know that there is another point $q\in P_{i,j+1}$ within distance
$\sqrt{d}R$. Applying this argument recursively and using that $P_{i,\ell+1} = \emptyset$, the lemma follows.
\end{proof}

\subsubsection{The Data Structure}

For every grid and every set $P_{i,j}$ we store the points using a hash table whose keys are the grid cells, i.e. a point $p$ is hashed to the bucket whose key is the grid cell that contains $p$.
For each bucket we maintain a second hash table and a doubly connected list that stores all points that fall into the same grid cell.
For each point in the hash table we maintain a pointer to its occurence in the list. This way, we can insert and delete points in
expected $O(d)$ time and we can return some point from the cell, if it is non-empty by returning the first list item.
We also store for each cell the point from $I_{i,j}$ it contains.

\subsubsection{Insertions}

We first argue how to implement insertions. Let $p$ be the point that is inserted. We start by inserting $P$ into the sets $P_{i,j}$. This
is simply done by choosing and storing a random number between $0$ and $1$. If this number is at most $n^{-j/\ell}$ then the point is
inserted into $P_{i,j}$. Note that if we condition on $p$ being in $P_{i,j}$ we have that the random value under this conditioning
is uniformly from $[0,n^{-j/\ell}]$. This implies that under this conditioning, $p$ in $P_{i,j+1}$ with probability $n^{-1/\ell}$.

We first update all hash tables by inserting $p$. Then we potentially need to update the sets $I_{i,j}$. 
We start by describing the insertion for the case that $p \in P_{i,j}$ but not in $P_{i,j+1}$. If $p$ is contained in a cell with a point from $P_{i,j+1}$ then we know that $p$ is covered
and we do not need further updates. Otherwise, we check whether $p$ is contained in a cell that contains another point from a set $I_{i,j}$. If this is the case, we are done. Otherwise,
we add $p$ to $I_{i,j}$ and remember for every grid cell of $p$ that it belongs to $I_{i,j}$ and we are done.

Now consider the case that $p$ is in $P_{i,j}$ and $P_{i,j+1}$. If for every cell of $p$ there is already a point $q\in P_{i,j+1}$ in the same cell, then we are done. Now consider any cell
that contains $p$ and that did not have a point from $P_{i,j+1}$ before. This means that all remaining points in $P_{i,j}$ in the same cell become covered and therefore we remove every point $q$
from $I_{i,j}$ that shares a cell with $p$. We then need to consider all points in the cells that contain $q$ and that do not contain a point from $P_{i,j+1}$ and for each point check whether it has to be inserted into $I_{i,j}$.
This can be done in the same way as in the first case (the points are not in $P_{i,j+1}$, which also means that the process does not cascade any more).

\begin{lemma}
Let $X_1,\dots, X_N$ be independent $0$-$1$-random variables with $\Pr{X_i=1} = \frac{1}{n^{1/\ell}}$. Then $\Pr{\sum_{i=1}^N X_i = 0} \le e^{-\frac{N}{2n^{1/\ell}}}$.
\end{lemma}
\begin{proof}
  Follows immediately from Chernoff bounds.
\end{proof}

\begin{corollary}
  Let $Q\subseteq P_{i,j}$ be a set of $N$ points inside the same grid cell. Consider an algorithm that processes $Q$ in time $O(1)$, if a point from $Q$
  is in $P_{i,j+1}$ and $\beta N$ otherwise. Then the expected running time to process $Q$ is $O(\beta \min\{N,n^{1/\ell}\})$.
\end{corollary}

By the above corollary we obtain that the expected time for insertion in the case that $p$ is also in $P_{i,j+1}$ is $O(dn^{1/\ell})$ per grid cell. Since the probability for
a point to be contained in $P_{i,j+1}$ is $n^{-1/\ell}$, the expected running time of inserting the point in one level is $O(dg)$ and the time to insert it into the data structure
for $P_i$ is $O(d g \ell)$.

\subsubsection{Deletions}

Now let us consider the case that $p$ is deleted. Consider the case that $p \in P_{i,j}$ but not in $P_{i,j+1}$. If $p$ is not in $I_{i,j}$ we can simply delete it from all hash tables
and we are done. Otherwise, we delete it and we need to update all other points that share a cell with $p$. This requires $O(d n^{1/\ell} g)$ time in expectation.
If $p$ is also in $P_{i,j+1}$ we consider all cells that contain $p$. If a cell does not contain another point from $P_{i,j+1}$ we update all remaining points from this cell, otherwise, we do not need to update the points.
This takes $O(dn^{1/\ell}g^2)$ time in expectation, because there are at most $g$ cells from which we need to process $O(n^{1/\ell})$ points in expectation.

\subsubsection{Finding Pointers between the \texorpdfstring{$P_i$}{Pi} and \texorpdfstring{$P_{i+1}$}{Pi+1}}

In contrast to the low-dimensional case, we do not store pointers between $P_i$ and $P_{i+1}$ explicitly. Instead, we show that we
can compute in $O(d g \ell)$ time such a pointer from our data structure. This can be done as follows for a point $p\in P_i= P_{i+1,0}$.
We first check whether
$p\in P_{i+1}$. This can be done by querying all grid cells that contain $p$ for all $P_{i+1,j}$ in time $O(dg\ell)$.
If $p \in P_{i+1}$ we have found our pointer. If this is not the case, we check for all grids whether $p$ is in the same grid cell as a
point $q$ from $P_{i+1,1}$. If this is the case, we know that the distance between $p$ and $q$
is at most $\sqrt{d} 2^{i+1}$. We do not know whether $q \in P_{i+1}$, so we apply this procedure recursively until we find a
point in $P_{i+1}$. This point has distance at most $\sqrt{d} \ell 2^{i+1}$ from $p$ by the triangle inequality.
Finally, if there does not exist a point $q \in P_{i+1}$ a point in the same grid cell as $p$ then either $p\in P_{i+1}$ or there
exists a point $q\in P_{i+1}$ in the same grid cell as $p$. Then $q$ has distance at most $\sqrt{d} 2^{i+1}$ and we found our pointer.
Overall, the time to find the pointer is $O(dg\ell)$. 

\thmHighDim*
\begin{proof}
  The success probability depends only on the question whether every pair at distance at most $2^i$ is contained in the same grid cell for at least one of the grids. We know
  that the probability of this event is $\frac{1}{2^g}$, so for $g\ge 10 \log n$ we get that this is simultanuously true for all $n^2$ pairs with probability at least $1-1/n$.
  The approximation guarantee follows similarly to the low-dimensional case. The expected insertion time follows from the fact that our data structure has $O(d \log \Delta)$ levels
  and the previous discussions.
\end{proof}

\bibliographystyle{amsalpha}
\bibliography{references}

\newcommand{\etalchar}[1]{$^{#1}$}
\providecommand{\bysame}{\leavevmode\hbox to3em{\hrulefill}\thinspace}
\providecommand{\MR}{\relax\ifhmode\unskip\space\fi MR }
\providecommand{\MRhref}[2]{%
  \href{http://www.ams.org/mathscinet-getitem?mr=#1}{#2}
}
\providecommand{\href}[2]{#2}
\begin{thebibliography}{DKM{\etalchar{+}}94}

\bibitem[ABC{\etalchar{+}}15]{ABCGMS15}
Hyung{-}Chan An, Aditya Bhaskara, Chandra Chekuri, Shalmoli Gupta, Vivek Madan,
  and Ola Svensson, \emph{Centrality of trees for capacitated k-center},
  Mathematical Programming \textbf{154} (2015), no.~1-2, 29--53.

\bibitem[BKP93]{BKP93}
Judit Bar{-}Ilan, Guy Kortsarz, and David Peleg, \emph{How to allocate network
  centers}, Journal of Algorithms \textbf{15} (1993), no.~3, 385--415.

\bibitem[CGK16]{CGK16}
Deeparnab Chakrabarty, Prachi Goyal, and Ravishankar Krishnaswamy, \emph{The
  non-uniform {$k$}-center problem}, Proceedings of the 43rd {I}nternational
  {C}olloquium on {A}utomata, {L}anguages, and {P}rogramming (ICALP), vol.~55,
  2016, pp.~Art. No. 67, 15.

\bibitem[CGS18]{CGS18}
T.{-}H.~Hubert Chan, Arnaud Guerqin, and Mauro Sozio, \emph{Fully dynamic
  \emph{k}-center clustering}, Proceedings of the 2018 World Wide Web
  Conference on World Wide Web (WWW), 2018, pp.~579--587.

\bibitem[CKLV17]{CKLV17}
Flavio Chierichetti, Ravi Kumar, Silvio Lattanzi, and Sergei Vassilvitskii,
  \emph{Fair clustering through fairlets}, Advances in Neural Information
  Processing Systems 30: Annual Conference on Neural Information Processing
  Systems 2017 (NIPS), 2017, pp.~5036--5044.

\bibitem[CLLW16]{CLLW16}
Danny~Z. Chen, Jian Li, Hongyu Liang, and Haitao Wang, \emph{Matroid and
  knapsack center problems}, Algorithmica \textbf{75} (2016), no.~1, 27--52.

\bibitem[CM15]{CM15}
Michael Cochez and Hao Mou, \emph{Twister tries: Approximate hierarchical
  agglomerative clustering for average distance in linear time}, Proceedings of
  the 2015 {ACM} {SIGMOD} International Conference on Management of Data, 2015,
  pp.~505--517.

\bibitem[CSS16]{CASS16}
Vincent Cohen{-}Addad, Chris Schwiegelshohn, and Christian Sohler,
  \emph{Diameter and k-center in sliding windows}, Proceedings of the 43rd
  International Colloquium on Automata, Languages, and Programming (ICALP),
  2016, pp.~19:1--19:12.

\bibitem[DKM{\etalchar{+}}94]{DKMMRT94}
Martin Dietzfelbinger, Anna~R. Karlin, Kurt Mehlhorn, Friedhelm {Meyer auf der
  Heide}, Hans Rohnert, and Robert~Endre Tarjan, \emph{Dynamic perfect hashing:
  Upper and lower bounds}, SIAM Journal on Computing \textbf{23} (1994), no.~4,
  738 -- 761.

\bibitem[DL05]{DL05}
Sanjoy Dasgupta and Philip~M. Long, \emph{Performance guarantees for
  hierarchical clustering}, Journal of Computer and System Sciences \textbf{70}
  (2005), no.~4, 555--569.

\bibitem[Epp00]{E00}
David Eppstein, \emph{Fast hierarchical clustering and other applications of
  dynamic closest pairs}, {ACM} Journal of Experimental Algorithmics \textbf{5}
  (2000), 1.

\bibitem[FG88]{FG88}
Tom{\'{a}}s Feder and Daniel~H. Greene, \emph{Optimal algorithms for
  approximate clustering}, Proceedings of the 20th Annual {ACM} Symposium on
  Theory of Computing (STOC), 1988, pp.~434--444.

\bibitem[Gon85]{G85}
Teofilo~F. Gonzalez, \emph{Clustering to minimize the maximum intercluster
  distance}, Theoretical Computer Science \textbf{38} (1985), 293--306.

\bibitem[GQD13]{GQD13}
Sean Gilpin, Buyue Qian, and Ian Davidson, \emph{Efficient hierarchical
  clustering of large high dimensional datasets}, 22nd {ACM} International
  Conference on Information and Knowledge Management (CIKM), 2013,
  pp.~1371--1380.

\bibitem[HN79]{HN79}
Wen{-}Lian Hsu and George~L. Nemhauser, \emph{Easy and hard bottleneck location
  problems}, Discrete Applied Mathematics \textbf{1} (1979), no.~3, 209--215.

\bibitem[HS86]{HS86}
Dorit~S. Hochbaum and David~B. Shmoys, \emph{A unified approach to
  approximation algorithms for bottleneck problems}, Journal of the {ACM}
  \textbf{33} (1986), no.~3, 533--550.

\bibitem[KPS00]{KPS00}
Samir Khuller, Robert Pless, and Yoram~J. Sussmann, \emph{Fault tolerant
  k-center problems}, Theoretical Computer Science \textbf{242} (2000),
  no.~1-2, 237--245.

\bibitem[KS00]{KS00}
Samir Khuller and Yoram~J. Sussmann, \emph{The capacitated \emph{K}-center
  problem}, {SIAM} Journal on Discrete Mathematics \textbf{13} (2000), no.~3,
  403--418.

\bibitem[LNRW10]{LNRW10}
Guolong Lin, Chandrashekhar Nagarajan, Rajmohan Rajaraman, and David~P.
  Williamson, \emph{A general approach for incremental approximation and
  hierarchical clustering}, {SIAM} Journal on Computing \textbf{39} (2010),
  no.~8, 3633--3669.

\bibitem[LYZ10]{LYZ10}
Jian Li, Ke~Yi, and Qin Zhang, \emph{Clustering with diversity}, Proceedings of
  the 37th International Colloquium on Automata, Languages and Programming
  (ICALP), 2010, pp.~188--200.

\bibitem[MK08]{MK08}
Richard~Matthew McCutchen and Samir Khuller, \emph{Streaming algorithms for
  k-center clustering with outliers and with anonymity}, Proceedings of the
  11th {APPROX} and 12th {RANDOM}, 2008, pp.~165--178.

\bibitem[PR04]{PR04}
Rasmus Pagh and Flemming~Friche Rodler, \emph{Cuckoo hashing}, Journal of
  Algorithms \textbf{51} (2004), no.~2, 122--144.

\end{thebibliography}

\end{document}